%% file: paper.tex
\pdfoutput=1
\newif\ifFull
\Fullfalse
\ifFull
\documentclass[11pt]{llncs}
\else
\documentclass{llncs}
\fi

\usepackage{graphicx}
\usepackage{amsfonts}
\usepackage{url}
\usepackage{times}
\usepackage{epstopdf}
\usepackage{subfigure}
\usepackage{boxedminipage}
\usepackage{amsmath}
\usepackage{algorithmic}
\usepackage[lined,boxed]{algorithm2e}
\usepackage[utf8]{inputenc}
\usepackage{hyperref}
\let\doendproof\endproof
\renewcommand\endproof{~\hfill\qed\doendproof}

\newtheorem{observation}[theorem]{Observation}

\begin{document}

\title{%
Wear Minimization for Cuckoo Hashing:
How~Not~to~Throw~a~Lot~of~Eggs~into~One~Basket}

\author{David Eppstein\inst{1} \and Michael T. Goodrich\inst{1}
\and Michael Mitzenmacher\inst{2} \and Pawe\l{} Pszona\inst{1}}

\institute{
Department of Computer Science,
University of California, Irvine
\and
School of Engineering and Applied Sciences,
Harvard University
}
\date{}

\maketitle

\input{abstract.tex}

% \ifFull\else
% \pagestyle{plain}
% \fi

\begin{quote}
``I did throw a lot of eggs into one basket, as you do in 
your teenage years...''\\[-2ex]

~\hfill---Dylan Moran~\cite{moran}
\end{quote}

\section{Introduction}

\input{intro.tex}

\section{Algorithm}

\input{algorithm.tex}

\section{Analysis}

\input{analysis.tex}

\section{Experiments}

\input{experiments.tex}

\subsection*{Acknowledgements}
This research was supported in part by NSF grants 1011840, 1228639, and
ONR grant N00014-08-1-1015.  Michael Mitzenmacher was supported in part
by NSF grants CCF-1320231, CNS-1228598, IIS-0964473, and CCF-0915922, and
part of this work was done while visiting Microsoft Research New England.

\raggedright
\bibliographystyle{abuser}
\bibliography{refs}

\end{document}

%% file: abstract.tex
% !TEX root = paper.tex
\begin{abstract}
We study wear-leveling techniques for cuckoo hashing, showing
that it is possible to achieve a memory wear
bound of $\log\log n+O(1)$ after the insertion of $n$ items
into a table of size $Cn$ for a suitable constant $C$ using cuckoo hashing.
Moreover, we study 
our cuckoo hashing method empirically, showing that it significantly
improves on the memory wear performance for classic cuckoo hashing 
and linear probing in practice.
\end{abstract}

%% file: intro.tex
% !TEX root = paper.tex
A \emph{dictionary}, or \emph{lookup table}, 
is a data structure for maintaining a set, $S$,
of key-value pairs, which are often called \emph{items}, to
support the following operations:
\begin{itemize}
  \item {\bf add}$(k,v)$: insert the item $(k,v)$ into $S$.
  \item {\bf remove}$(k)$: remove the item with key $k$ from $S$.
  \item {\bf lookup}$(k)$: return the value $v$ associated with $k$ in $S$ 
(or return \textbf{null} if there is no item with key equal to $k$ in $S$).
\end{itemize}
The best dictionaries, both in theory and in practice, are implemented
using hashing techniques (e.g., see~\cite{bp-s-10,Knuth:1998}),
which typically achieve $O(1)$ expected time performance for 
each of the above operations.

One technique that has garnered much attention of late is
\emph{cuckoo hashing}~\cite{pagh2004cuckoo}, which is a hashing method
that supports lookup and remove operations in 
\emph{worst-case} $O(1)$ time, with
insertions running in expected $O(1)$ time.
Based on the metaphor of how the European cuckoo bird lays its eggs, 
inserting an item via cuckoo hashing involves throwing the item into 
one of two possible cells and evicting any existing item, which in turn
moves to its other cell, repeating the process.
In this standard formulation,
each item in $S$ is stored in one of two possible
locations in a table, $T$, but other variants have also been
studied as well
(e.g., see~\cite{DBLP:conf/icalp/ArbitmanNS09,devroye2003cuckoo,DBLP:conf/icalp/DietzfelbingerGMMPR10,DBLP:conf/approx/FriezeMM09,DBLP:journals/siamcomp/KirschMW09}). For instance, the \emph{cuckoo hashing with a stash} variation sacrifices some of the elegance of the basic algorithm by adding an auxiliary 
cache to break insertion cycles, greatly reducing the failure probability of the algorithm and reducing the need to rebuild the structure when a failure happens~\cite{DBLP:journals/siamcomp/KirschMW09}.
There have also been improvements to the space complexity for cuckoo 
hashing~\cite{Fotakis,DBLP:conf/esa/LehmanP09},
as well as experimental studies of
the practical applications of 
cuckoo hashing~\cite{Alcantara:2009:RPH:1661412.1618500,Zukowski:2006:AH:1140402.1140410}.
In spite of its elegance, efficiency, and practicality, however,
cuckoo hashing has a major drawback with respect to modern memory technologies.

\subsection{Wear Leveling in Modern Memories}
Phase-change memory (PCM) and flash memory are memory technologies
gaining in modern interest,
due to their persistence and energy efficiency properties. 
For instance, they are being
used for main memories as well as for external storage
(e.g., see~\cite{bcmv-ifm-03,pboz-fmc-97,wet-pcm-10,Wu:1994}).
Unfortunately, such modern memory technologies suffer from \emph{memory wear},
a phenomenon caused by charge trapping in the gate
oxide~\cite{Grupp:2009:CFM:1669112.1669118}, in which extensive writes
to the same cell in such memories
(typically 10,000--100,000 write/erase cycles for flash memory) 
causes all memory cells to become permanently unusable.
% (although there is active research on 
% self-healing flash memory~\cite{6479008}). 

To deal with this drawback,
several heuristic techniques for \emph{wear leveling} have been
proposed and
studied
(e.g., see~\cite{bt-cafm-06,chang2007efficient,chang2007endurance,%
chk-ifw-10,UBIFS,Qureshi:2009,woodhouse2001jffs}).
Such techniques are intended to limit memory wear, but they do not produce 
high-probability wear guarantees; hence, we are motivated to provide
special-purpose solutions for common data structure applications, such as 
for dictionaries, that have provable wear-leveling guarantees.

In this paper, we focus on a previously-ignored aspect of cuckoo
hashing, namely its \emph{maximum wear}---the maximum number of writes
to any cell in the hash table, $T$, during the execution 
of a sequence of $n$ operations.
Unfortunately, in spite of their other nice properties, cuckoo hashing
and its variants do not have good wear-leveling properties.
For instance, after inserting $n$ items
into a hash table, $T$, implemented using cuckoo hashing,
the expected maximum wear of a cell in $T$ 
is $\Omega(\log n/\log \log n)$---this follows from
the well-known balls-in-bins result~\cite{mitzenmacher2005probability}
that throwing $n$ balls into $n$ bins results in the expected
maximum number of balls in a bin being $\Theta(\log n/\log\log n)$.
Ideally, we would strongly prefer there to be 
a simple way to implement cuckoo hashing
that causes its memory wear to be closer to $O(1)$.

% In this context, minimizing the maximum number of writes to any cell of the
% hash table is equivalent to maximizing the time of the first failure when
% cuckoo hashing is used as a wear leveling technique.
We study a simple but previously unstudied modification to cuckoo hashing that uses $d > 2$ hash functions, and that chooses where to insert or reinsert each item based on the wear counts of the cells it hashes to.
We prove that this method achieves, with high probability, 
maximum wear of $\log\log n +O(1)$
after a sequence of $n$ insertions into a hash table of size $Cn$, where $C$
is a small constant. 
Moreover, we show experimentally that this variant 
achieves significantly reduced wear compared to classical cuckoo hashing
in practice, with deletions as well as insertions.

\subsection{Related Work}
With respect to previous algorithmic work on memory wear leveling,
Ben-Aroya and Toledo~\cite{bt-cafm-06}
introduce a memory-wear model
consisting of $N\ge2$ cells, indexed
from $0$ to $N-1$, such that each is a single memory word or block,
and such that there is a known parameter,
$L$, specifying an upper-bound limit on the number of 
times that any cell of this memory can be rewritten.
They study competitive ratios of 
wear-leveling heuristics from an online 
algorithms perspective.
In addition, in even earlier work,
Irani {\it et al.}~\cite{inr-otscc-92} study several schemes for performing
general and specific algorithms using write-once memories.
Qureshi {\it et al.}~\cite{Qureshi:2009}, on the other hand, describe a
wear-leveling strategy, called the \emph{Start-Gap} method, which 
works well in practice but has not yet been analyzed theoretically.
With respect to other previous work involving
wear-leveling analysis of specific algorithms and data structures, 
Chen {\it et al.}~\cite{ChenGN11}
study the wear-leveling performance of methods for 
performing database index and join operations.

Azar {\it et al.}~\cite{abku-ba-99} show that if one throws $n$
balls into $n$ bins, but with each ball being added to the least-full bin from
$k\ge 2$ random choices, 
then the largest bin will have size $\log\log n+O(1)$ with 
high probability. 
Subsequentially, other researchers have discovered further applications
of this approach, which has given rise to a paradigm
known as the \emph{power of two choices} (e.g., see~\cite{mrs-p2-01}).
This paradigm is exploited further in the work on
\emph{cuckoo hashing} by Pagh and Rodler~\cite{pagh2004cuckoo},
which, as mentioned above, uses two random hash locations for items
along with the ability to dynamically move keys to any of its hash
locations so as to support worst-case $O(1)$-time lookups and removals.
Several variants of cuckoo hashing have been considered, including
the use of a small cache (called a ``stash'') and the use of more than two
hash functions
(e.g., see~\cite{DBLP:conf/icalp/ArbitmanNS09,devroye2003cuckoo,DBLP:conf/icalp/DietzfelbingerGMMPR10,DBLP:conf/approx/FriezeMM09,DBLP:journals/siamcomp/KirschMW09}). 
Also, as mentioned above, there has also been work
improving the space complexity~\cite{Fotakis,DBLP:conf/esa/LehmanP09},
as well as experimental analyses of
cuckoo hashing~\cite{Alcantara:2009:RPH:1661412.1618500,Zukowski:2006:AH:1140402.1140410}.
Nevertheless, we are not familiar with any previous work on the 
memory-wear performance of cuckoo hashing, from either a theoretical
or experimental viewpoint.

Our approach to improving the memory wear for cuckoo hashing is to use a 
technique that could be called the \emph{power of three choices}, in that
we exploit the additional 
freedom allowed by implementing cuckoo hashing with at least three
hash functions instead of two.
Of course, as cited earlier, previous 
researchers have considered cuckoo hashing with more than two hash functions
and the balls-in-bins 
analysis of Azar {\it et al.}~\cite{abku-ba-99} applies to any number
of at least two choices.
Nevertheless,
previous many-choice cuckoo hashing schemes do not have good memory wear bounds,
since, even by the very
first random choice for each item, there is an $\Omega(\log n/\log\log n)$ 
expected maximum memory-wear bound for classic cuckoo hashing with $d\ge3$ 
hash functions.
In addition, the approach of Azar {\it et al.}~does not seem to lead to
a bound of $O(\log\log n)$
for the memory wear of the variant of cuckoo hashing we consider,
because items in a cuckoo hashing scheme can move to any of their
alternative hash locations during
insertions, rather than staying in their ``bins'' as required in the framework
of Azar {\it et al.}
Moreover, there are non-trivial dependencies that exist between 
the locations where an item is moved and a sequence of insertions that
caused those moves, which complicates any theoretical analysis.

\subsection{Our Results}
In this paper, we consider the {\em memory wear} of cells of a cuckoo hash
table, where by ``wear'' we mean the number of times a cell is
rewritten with a new value.
We introduce a new
cuckoo hashing insertion rule that essentially involves determining the next 
location to ``throw'' an item 
by considering the existing wear of 
the possible choices.
We provide a theoretical analysis proving that our simple
rule has low wear---namely, for a
suitably small constant $\alpha_d$, when inserting $\alpha n$ items
into a cuckoo hash table with $n$ cells, $d$ distinct choices per
item, and $\alpha < \alpha_d$, the maximum wear on any cell can be bounded
by $\log \log n +O(1)$.

Of course, for any reasonable values of $n$, the above $\log\log n$
term is essentially constant; hence, our theoretical analysis necessarily
begs an interesting question: 
\begin{quotation}
\noindent
\textbf{
Does our modified insertion rule for cuckoo hashing lead to improved memory-wear
performance in practice?
}
\end{quotation}
To answer this question, we performed a suite of experiements which 
show that, in fact, even for large, but non-astronomical, values of $n$,
our simple insertion rule leads to significant improvements
in the memory-wear performance of cuckoo hashing over classic cuckoo hashing
and hashing via linear probing.

%% file: algorithm.tex
% !TEX root = paper.tex

In classical cuckoo hashing, each item of the dynamic set $S$ is hashed to two possible cells of the hash table; these may either be part of a single large table or two separate tables. When an item $x$ is inserted into the table (at the first of its two cells), it displaces any item $y$ already stored in that cell. The displaced item $y$ is reinserted at its other cell, which in turn may displace a third item $z$, and so on.  If $n$ items are inserted into a table whose number of cells is a sufficiently large multiple of $n$, then with high probability all chains of displaced items eventually terminate within $O(\log n)$ steps, producing a data structure in which each lookup operation may be performed in constant time. The expected time per insertion is $O(1)$. Removals may be performed in the trivial way, by simply removing an item from its cell and leaving all other items in their places.  By a standard analysis of balls-and-bins processes, the wear arising just from placing items into their first cells, not even counting the additional wear from re-placing displaced items, is $\Omega(\log n/\log\log n)$.  

A standard variant of cuckoo hashing allows more than two choices.  In this setting, when an item $x$ is inserted, it first goes through its choices in order to see if any are empty.  If none are, one must be displaced.  A common approach in this case is to use what is called random walk cuckoo hashing \cite{Fotakis,Frieze:2011:ARC:2078866.2078869};  when an item is to be displaced, it is chosen randomly from an item's choices.  (Usually, after the first displacement, one does not allow the item that was immediately just placed to be displaced again, as this would be wasteful in terms of moves.)  
Allowing more choices allows for higher load factors in the hash table and makes the probability of a failure much smaller.

We modify the standard algorithm, using $d\ge 3$ choices for each item, in the following ways:
\begin{itemize}
\item Associated with each cell we store a \emph{wear count} of the number of times that cell has been overwritten by our structure.
\item When an item is inserted, it is placed into one of its $d$ cells with the lowest wear count, rather than always being placed into the first of its cells.
\item When an item $y$ is displaced from a cell, we compare the wear counts of all $d$ of the cells associated with $y$ (including the one it has just been displaced from, causing its wear count to increase by one). We then store $y$ in the cell with the minimum wear count, breaking ties arbitrarily, and displace any item that might already be placed in this cell.
\end{itemize}

A computational shortcut that makes no difference to the analysis of our algorithm is to take note of a situation in which two items $x$ and $y$ are repeatedly displacing each other from the same cell; in this case, the repeated displacement will terminate when the wear count of the cell reaches a number at least as high as the next smallest wear count of one of the other cells of $x$ or $y$. Rather than explicitly performing the repeated displacement, the algorithm may calculate the wear count at which this termination will happen, and perform the whole sequence of repeated displacements in constant time.

%% file: analysis.tex
% !TEX root = paper.tex
As is common in theoretical analyses of cuckoo hashing
(e.g., see~\cite{%
DBLP:conf/icalp/ArbitmanNS09,devroye2003cuckoo,DBLP:conf/icalp/DietzfelbingerGMMPR10,DBLP:conf/approx/FriezeMM09,DBLP:journals/siamcomp/KirschMW09,pagh2004cuckoo}),
let us view the cuckoo hash table in terms of a random
hypergraph, $H$, where there is a vertex in $H$
for each cell in the cuckoo hash table, and there is a hyperedge
in $H$
for each inserted item, with the endpoints of this hyperedge determined
by the $d$ cells associated with this item.
For the sake of this analysis, 
we consider a sequence of $\alpha n$ insertions into an initially
empty table of size, $n$, and we disallow deletions.
(Although the analysis of 
mixed sequences of insertions and deletions 
is beyond the approach of this section, we study such mixed 
update sequences experimentally later in this paper.)
Thus, $H$ has $n$ vertices and $\alpha n$ hyperedges, each of which
is a subset of $d$ vertices.
In this context, we say that a subgraph of $H$ with $s$ edges and $r$ vertices
is a \emph{tree}, if $(d-1)s=r-1$,
and it is \emph{unicyclic}, if $(d-1)s=r$
(e.g., see~\cite{KarLuc-JCAM-02}).
Since the $d$ hash functions are random, $H$ is a standard random $d$-ary 
(or ``$d$-uniform'') hypergraph, with $d\ge 3$ being a fixed constant;
hence, we may use the following
well-known fact about random hypergraphs.

\begin{lemma} [Schmidt-Pruzan and Shamir~\cite{SchSha-Comb-85}; Karo{\'n}ski and {\L}uczak~\cite{KarLuc-JCAM-02}]
\label{lem:log-components}
When $\alpha < 1/(d(d-1))$, with probability
$1 - O(1/n)$ the maximum connected component size in $H$ is $O(\log n)$, and
all components are either trees or unicyclic components.  
\end{lemma}

Here the constant implicit in the $O(\log n)$ may depend on $\alpha$ and $d$,
and affects the $O(1)$ term in the final $\log \log n +O(1)$ bound on the wear.

We define an orientation on $H$ where
each hyperedge in $H$ is oriented to one of its $d$ cells, denoting where the 
associated item resides.  
Thus, such an orientation has at most one hyperedge oriented to
each vertex, since there is at most one item located in each cell.

Recall our modified cuckoo hashing insertion procedure, which
allows us to achieve our bound on wear.  
We assume that each cell
tracks its current wear (the number of times the cell value
has been rewritten) in a counter.  
On insertion for a time, $x$, 
if any of the $d$ cells for $x$ is empty,
then $x$ is placed in one of its empty cells (randomly chosen).  
Otherwise, $x$ replaces the item, $y$, 
in the cell with the smallest wear out of
its $d$ possible choices.  The replaced item, $y$, continues by replacing the
item in the cell with the smallest wear out of its $d$ choices, and so on,
until all items are placed.  Note that a pair of items, $x$ and $y$,
may repeatedly replace each other in a cell until the wear of the cell 
increases to the point when there is an alternative lower-wear cell for 
$x$ or $y$
(in practice we would not have perform the repeated replacements, but would 
update the wear variable accordingly).

For a cell of wear $k$, we define its \emph{wear-children} to be the $d-1$
other cells corresponding to the item contained in this cell.  
Our proof of the $O(\log \log n)$ bound on the 
maximum wear for our modified cuckoo hashing scheme depends, in part,
on the following three simple, but crucial, observations.

\begin{observation}
\label{obs:one}
A cell of wear $k$ has $d-1$ wear-children with wear at least $k-1$.  
\end{observation}
This is because for a cell to obtain wear $k$ all the choice corresponding
to the item placed in that cell have to have wear at least $k-1$ at the time
of its placement.  

\begin{figure}[t]
\centering\includegraphics[width=\textwidth]{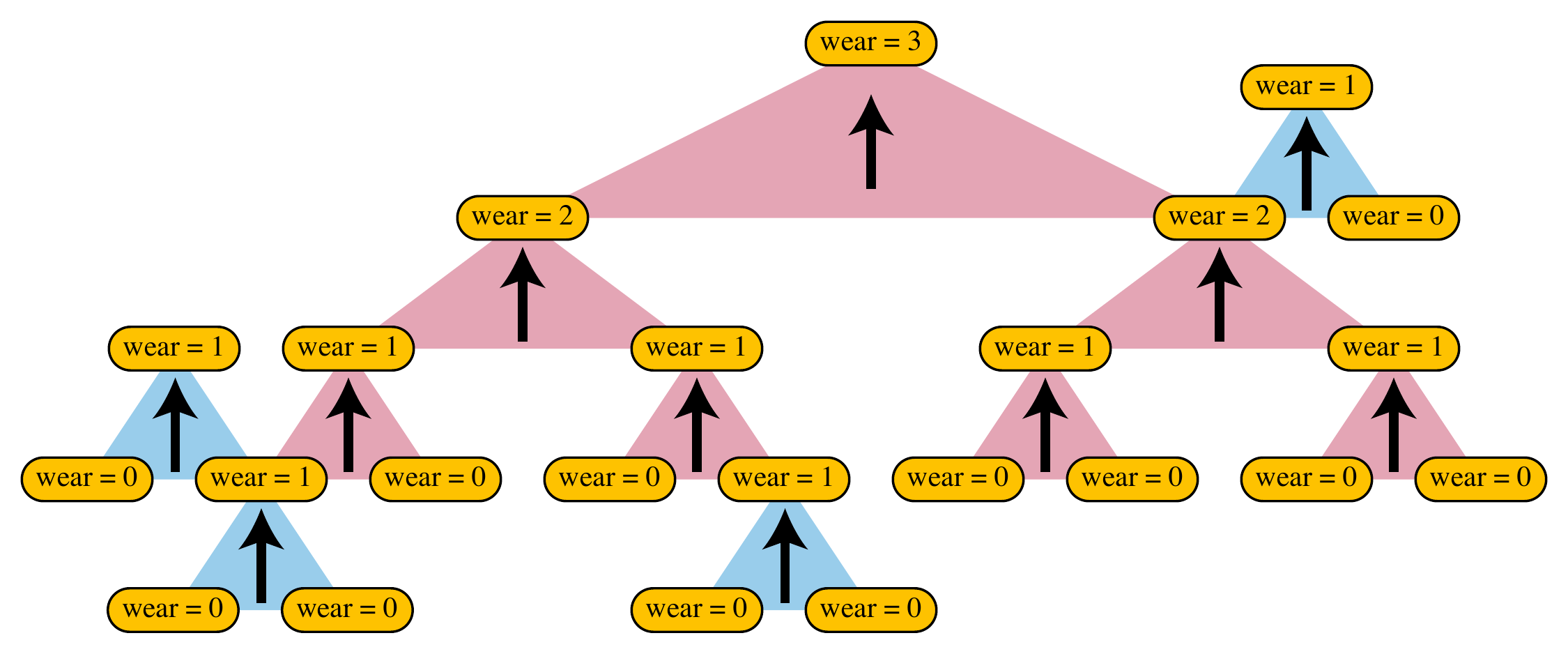}
\caption{A tree component in the hypergraph of cells (yellow ovals) and values (blue and pink triangles) for a cuckoo hash table with $d=3$. The arrows in the value triangles indicate the cell into which each value is placed. The complete binary tree associated with a cell of wear~$3$ is shown by the pink triangles. Some additional values and cells needed to achieve the depicted wear values are not shown.}
\label{fig:binary-tree}
\end{figure}

\begin{observation}
\label{obs:tree}
Consider any tree component in the hypergraph $H$ with a cell of wear $k$.  Then
the component contains a $(d-1)$-ary tree consisting of the cell's wear-children, and the
wear-children of those wear-children, and so on, of at least $(d-1)^k+1$ distinct nodes.  
\end{observation}
This is because, by Observation~\ref{obs:one},
a cell of wear $k$ has $d-1$  wear-children with wear at least
$k-1$, each of which has $d-1$ of its own children, and we can iterate 
this argument counting descendants down to wear 0.  
The nodes must be distinct when the component is a tree. 
Fig.~\ref{fig:binary-tree} depicts an example.

If we did not need to consider unicyclic components, we would be done
at this point. This is because, by
Lemma~\ref{lem:log-components},
component sizes are bounded by $O(\log n)$ with high probability, 
and, if all components are trees, then, by Observation~\ref{obs:tree}, 
if there is a cell of wear $k$, then
there is a component of size at least $(d-1)^k$. 
Thus, for $d \geq 3$,
the maximum wear would necessarily be $\log \log n +O(1)$ with high probability.
Our next observation, therefore, deals with unicyclic components.

\begin{observation}
\label{obs:unicycle}
Consider any unicyclic component in the hypergraph $H$ 
with a cell of wear $k$.  Then
the component contains a $(d-1)$-ary tree of $(d-1)^{k-1}+1$ distinct nodes.  
\end{observation}
\begin{proof}
We follow the same argument as for a tree component, in which we build a complete $(d-1)$-ary tree from the given cell. Because the component is not acyclic, this process might find more than one path to the same cell. If this should ever happen, we keep only the first instance of that cell, and prune the tree at the second instance of the same cell (preventing it from being a complete tree). This pruning step breaks the only cycle in a unicyclic component, so there is only one such pruning and the result is a tree that is complete except for one missing branch. The worst place for the pruning step to occur is nearest the root of the tree tree, in which case at most a $1/(d-1)$ fraction of the tree is cut off.    
\end{proof}

It is clear that unicyclic components do not change our conclusion that the maximum
wear is $\log \log n +O(1)$ with high probability.

\begin{theorem}
If our modified cuckoo hashing algorithm performs a sequence of $n$ insertions, with $d\ge 3$, on a table whose size is a sufficiently large multiple of $n$, then with high probability the wear will be $\log\log n+O(1)$.
\end{theorem}

\begin{proof}
We form a $d$-regular hypergraph whose vertices are the cells of the hash table, and whose hyperedges record the sets of cells associated with each item.
By Lemma~\ref{lem:log-components}, with high probability, all components of this graph are trees or unicyclic components, with $O(\log n)$ vertices. By Observations~\ref{obs:tree} and~\ref{obs:unicycle}, any cell with wear $k$ in one of these components has associated with it a binary tree with $\Omega(2^k)$ vertices. In order to satisfy $2^k=O(\log n)$ (the binary tree cannot be larger than the component containing it) we must have $k\le \log\log n+O(1)$.
\end{proof}

It would be of interest to extend this analysis to sequences of operations that mix insertions with deletions, but our current methods handle only insertions.

%% file: experiments.tex
% !TEX root = paper.tex
We have implemented and tested three hashing algorithms:
\begin{itemize}
  \item our variant of cuckoo hashing with $d=3$ hash functions
  \item standard cuckoo hashing with $d=3$ hash functions\footnote{%
	Using standard cuckoo hashing with $d=2$ proved counterproductive,
	as multiple \emph{failures} (i.e., situations where we are unable
	to insert new item into the table) were
	observed. This is due to the failure probability being non-negligible,
	of the order $\Theta(1/n)$, in this version of cuckoo hashing. 
        Of course, such failures can be circumvented by using a 
	\emph{stash}~\cite{DBLP:journals/siamcomp/KirschMW09}
	for storing items that failed to be inserted, 
	but a stash necessarily has to be 
	outside of the wear-vulnerable memory, 
	or it has to be moved a lot,
	since it will have many rewrites. Thus, we did not include comparisons
	to cuckoo hashing with two hash functions.}
  \item standard open-address hashing with linear 
probing and eager deletion (explained below)~\cite{Sed-03-linear-probing}.
\end{itemize}
We ran a series of tests to gauge the behavior of maximum wear for the above
three algorithms. 
In all cases, the setup was the same: we start with an initially
empty hash table of capacity 30 million items, 
then perform a number of insertions,
until desired usage ratio (fill ratio) is achieved 
(we tested usage ratios 1/6, 1/3, 1/2, 2/3 and 4/5).
Then, for the main part of the test, we perform 1 billion ($10^9$) operation
pairs, where a single operation pair consists of
\begin{enumerate}
  \item deleting a randomly selected item that is present in the hash table
  \item inserting a new item into the hash table.
\end{enumerate}
This way, once the desired usage ratio is achieved in the first phase, it is
kept constant throughout the rest of the test.
For the sequence of insertions, we simply used integers in the natural order
(i.e., 0, 1, 2, $\ldots$), since the hash functions are random. 
Different test runs with the same input were
parametrized by the use of different hash functions in the algorithm.
We ran about 13 tests for each 
(\emph{usage ratio}, \emph{algorithm}) combination.
Our tests were implemented in \verb\C++\. We used various cryptographic hash
functions provided by the \verb\mhash\ library~\cite{mhash}.

As discussed above,
we store wear information for each cell in the hash table. 
Each time a new value
is written into the cell, we increase its wear count. In the case of linear
probing, when an item, $a$, is erased, all subsequent items in the same
probe sequence (until an empty cell is encountered) need to be rehashed. This is
implemented as erasing an item and inserting it again using the insertion
algorithm. It is easy to implement this in a way such that if a rehashed item
ends up in its original position, no physical write is necessary. Therefore
we did not count this case as a wear increase.

In addition to the above scenario of mixed insertions and deletions,
we also measured the max wear after
an insertion of 20 million items into a hash
table of capacity 30 million, with no deletions.
The results are shown in Fig.~\ref{fig_insert}.

\begin{figure}
\vspace*{-0.3in}
\begin{center}
  \includegraphics[width=0.9\textwidth,height=0.35\textheight]{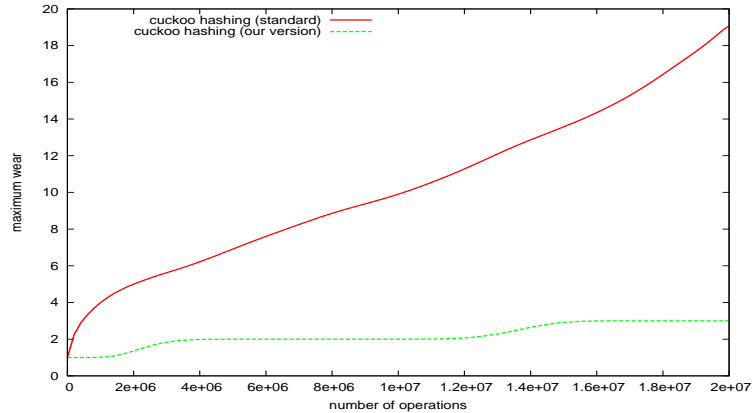}
\end{center}
\vspace*{-0.4in}
  \caption[cc7]{Average maximum wear during initial insertion 
    sequence (20 million insertions
    into hash table of capacity 30 million), with no deletions.
    Linear probing is not shown, as
    it has maximum wear equal to 1 as long as insertions are 
    the only operations involved.}
  \label{fig_insert}
\end{figure}

Even though our theoretical analysis is only valid for sequences of insertions,
our algorithm behaves equally well when insertions are interspersed
with deletions, as explained in the description of the test setting. 
The results for different usage ratios are shown in 
Figures~\ref{fig_1_6}, \ref{fig_1_3}, \ref{fig_1_2}, \ref{fig_2_3} and~\ref{fig_4_5}.
Table~\ref{avg_wear_table} contains average wear (calculated over all cells)
when the test has concluded.

\begin{figure}
\vspace*{-0.2in}
\begin{center}
  \includegraphics[width=0.9\textwidth,height=0.35\textheight]{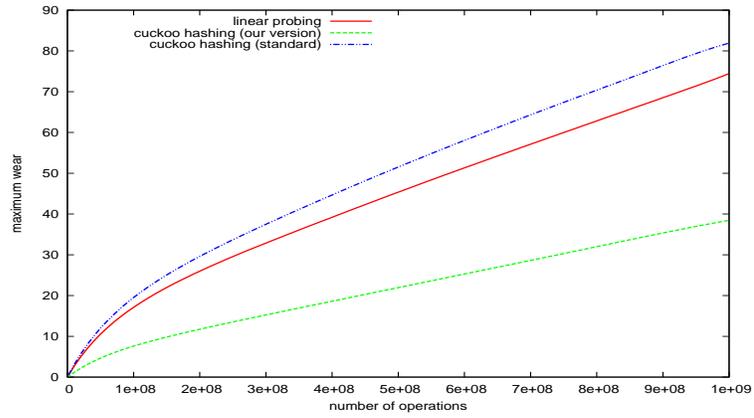}
\end{center}
\vspace*{-0.5in}
  \caption{Maximum wear for usage ratio $\frac{1}{6}$.}
  \label{fig_1_6}
\end{figure}

\begin{figure}
\vspace*{-0.2in}
\begin{center}
  \includegraphics[width=0.9\textwidth,height=0.35\textheight]{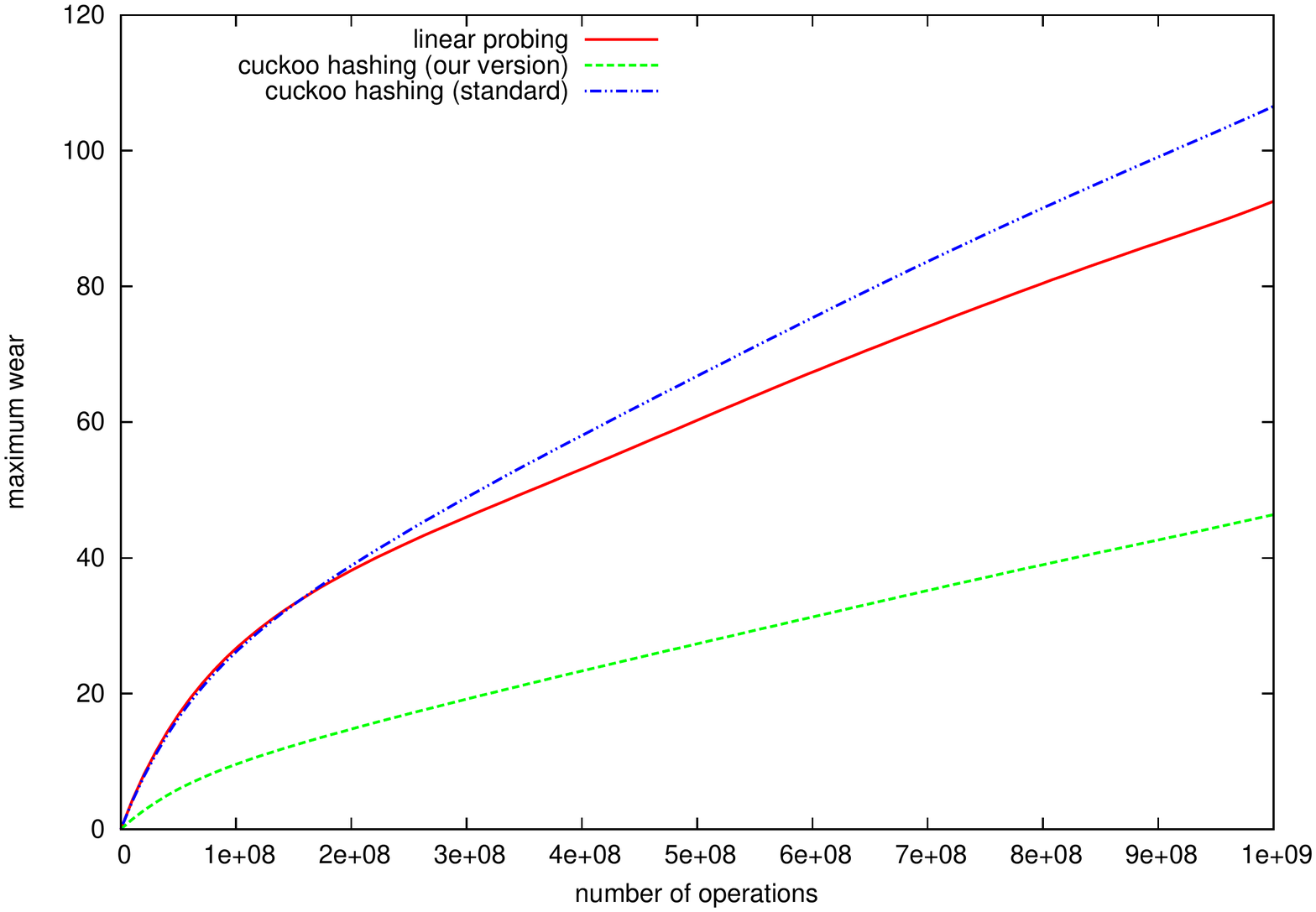}
\end{center}
\vspace*{-0.5in}
  \caption{Maximum wear for usage ratio $\frac{1}{3}$.}
  \label{fig_1_3}
\end{figure}

\begin{figure}
\vspace*{-0.2in}
\begin{center}
  \includegraphics[width=0.9\textwidth,height=0.35\textheight]{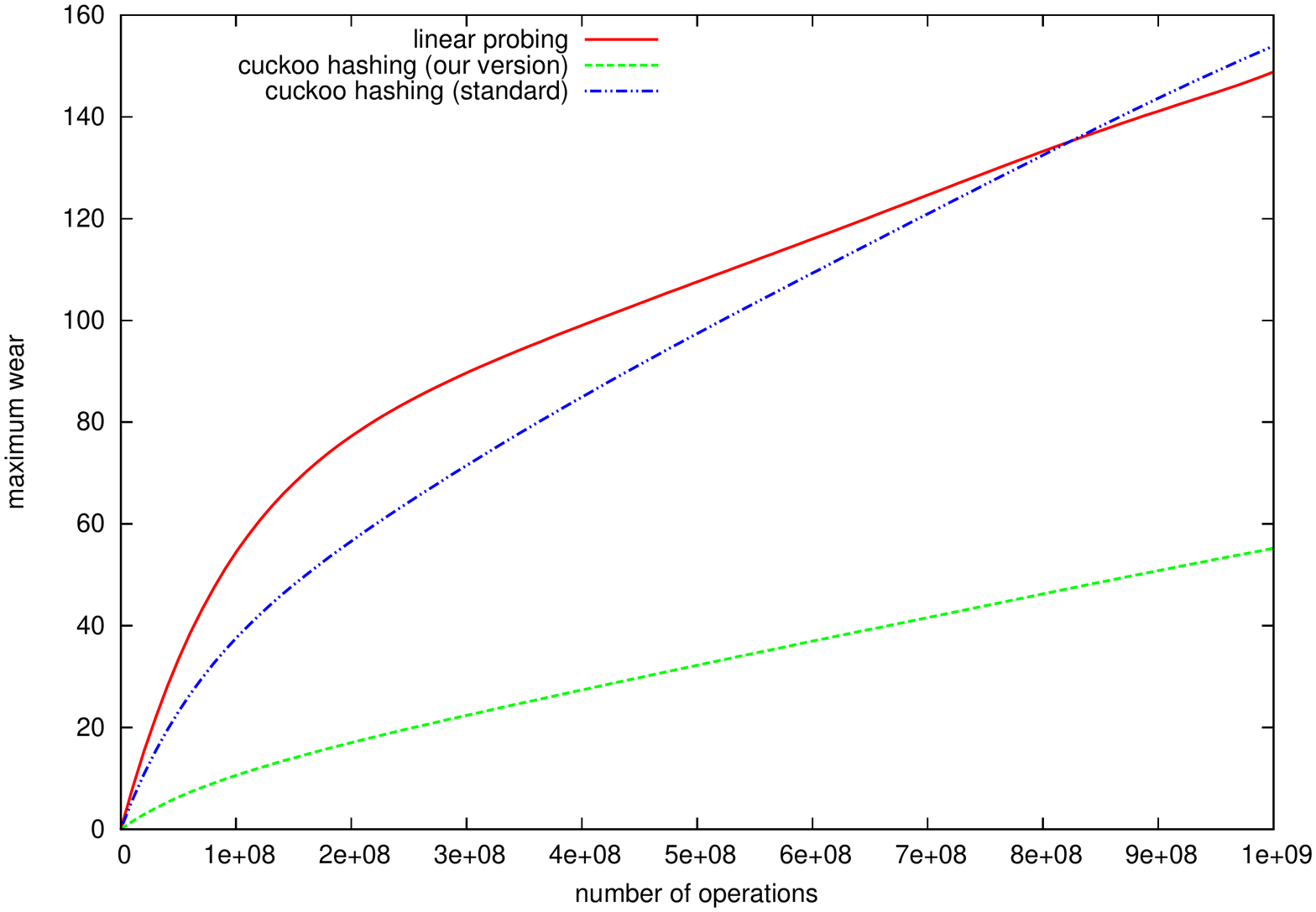}
\end{center}
\vspace*{-0.5in}
  \caption{Maximum wear for usage ratio $\frac{1}{2}$.}
  \label{fig_1_2}
\end{figure}

\begin{figure}
\vspace*{-0.2in}
\begin{center}
  \includegraphics[width=0.9\textwidth,height=0.35\textheight]{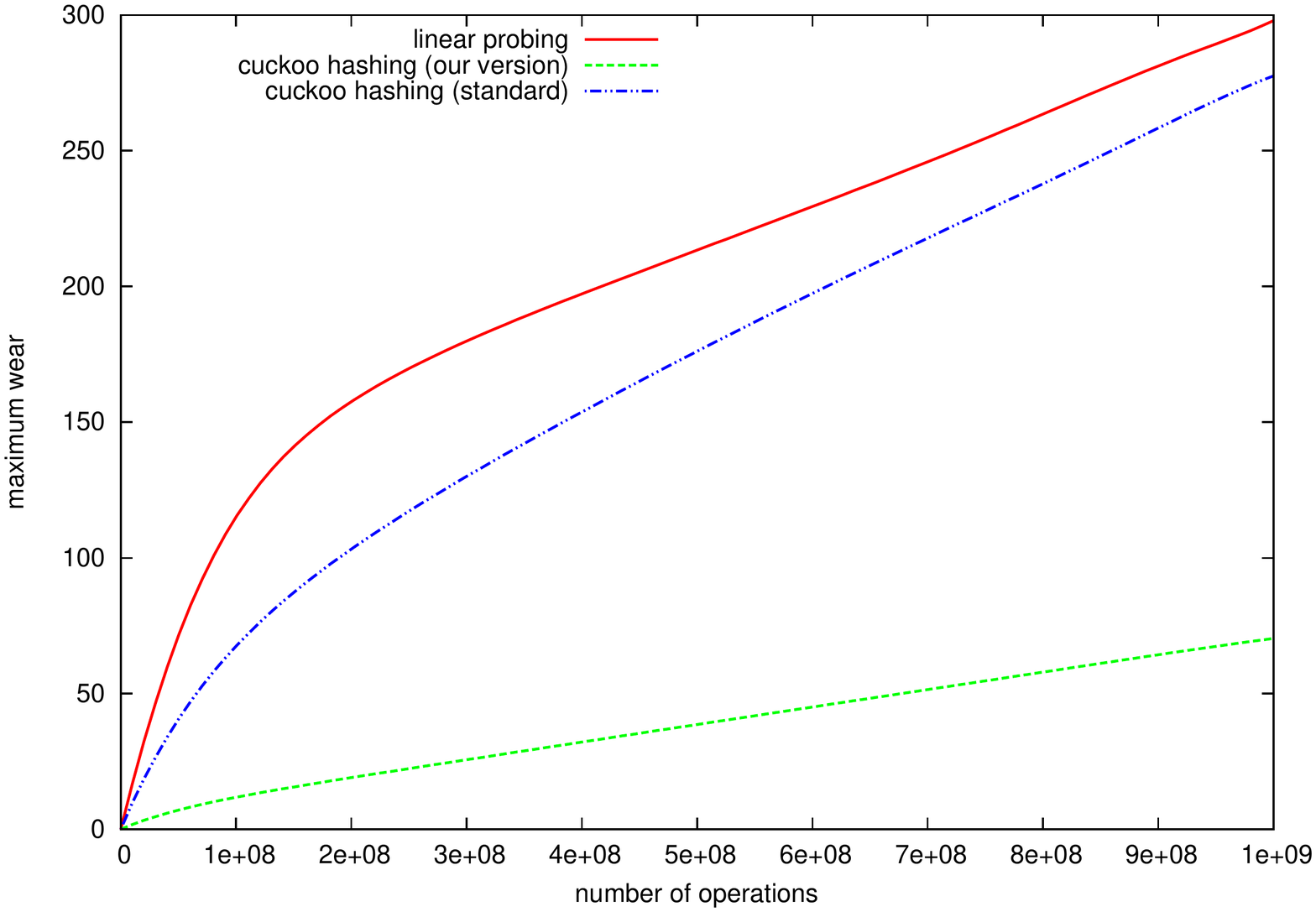}
\end{center}
\vspace*{-0.5in}
  \caption{Maximum wear for usage ratio $\frac{2}{3}$.}
  \label{fig_2_3}
\end{figure}

\begin{figure}
\vspace*{-0.2in}
\begin{center}
  \includegraphics[width=0.9\textwidth,height=0.35\textheight]{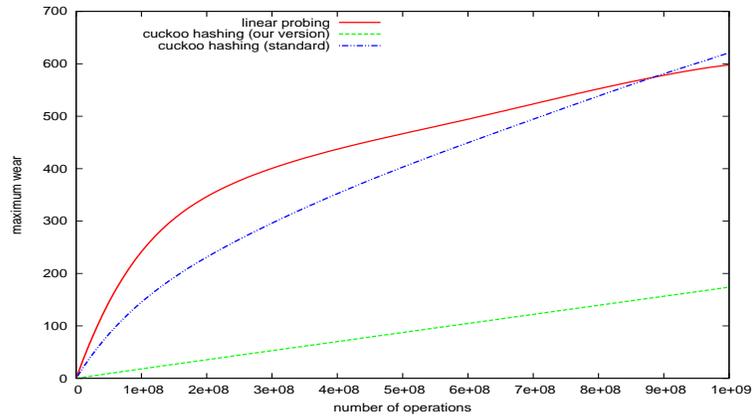}
\end{center}
\vspace*{-0.5in}
  \caption{Maximum wear for usage ratio $\frac{4}{5}$.}
  \label{fig_4_5}
\end{figure}

\begin{table}
  \begin{center}
    \begin{tabular}{|c|c|c|c|}
      \hline Usage ratio & Cuckoo hashing (ours) & Cuckoo hashing (standard) & Linear probing \\
      \hline 1/6 & 33.92 & 40.52 & 37.12 \\
      \hline 1/3 & 36.57 & 52.33 & 44.41 \\
      \hline 1/2 & 44.68 & 74.51 & 58.96 \\
      \hline 2/3 & 64.52 & 128.91 & 90.19\\
      \hline 4/5 & 171.93 & 281.84 & 148.29\\
      \hline
    \end{tabular}
  \end{center}
  \caption{Average wear after sequence of operations.}
  \label{avg_wear_table}
\end{table}

\subsection{Discussion} 
When the only operations performed are insertions,
linear probing achieves optimal wear (wear${}=1$). 
However, when a mixed sequence
of insertions and deletions is performed, our algorithm is clearly the best
among those tested in terms of minimizing maximum wear.
In this situation, linear probing and standard cuckoo hashing
behave in a similar way, while the maximum wear is much smaller 
for our algorithm.
It is evident that the difference grows as the hash table becomes more filled
(its usage ratio becomes higher).
It is also worth noticing that minimizing maximum wear also results in smaller
overall wear of the hash table, as shown in Table~\ref{avg_wear_table} (exept
for the case of wear${}=4/5$, where linear probing achieves smaller overall wear,
but at the cost of unacceptable slowdown in access time).

%% file: paper.bbl
\begin{thebibliography}{10}

\bibitem{Alcantara:2009:RPH:1661412.1618500}
D.~A. Alcantara, A.~Sharf, F.~Abbasinejad, S.~Sengupta, M.~Mitzenmacher, J.~D.
  Owens, and N.~Amenta.
\newblock {Real-time parallel hashing on the GPU}.
\newblock {\em Proc. ACM SIGGRAPH Asia 2009}, pp.~154:1{--}154:9, 2009,
  \href{http://dx.doi.org/10.1145/1661412.1618500}%
{doi:10.1145/1661412.1618500}.

\bibitem{DBLP:conf/icalp/ArbitmanNS09}
Y.~Arbitman, M.~Naor, and G.~Segev.
\newblock {De-amortized cuckoo hashing: Provable worst-case performance and
  experimental results}.
\newblock {\em Automata, Languages and Programming: 36th International
  Colloquium, ICALP 2009, Rhodes, Greece, July 5-12, 2009, Proceedings, Part
  I}, pp.~107{--}118. Springer, Lecture Notes in Computer Science 5555, 2009,
  \href{http://dx.doi.org/10.1007/978-3-642-02927-1\_11}%
{doi:10.1007/978-3-642-02927-1\_11}.

\bibitem{abku-ba-99}
Y.~Azar, A.~Broder, A.~Karlin, and E.~Upfal.
\newblock {Balanced allocations}.
\newblock {\em SIAM J. Comput.} 29(1):180{--}200, 1999,
  \href{http://dx.doi.org/10.1137/S0097539795288490}%
{doi:10.1137/S0097539795288490}.

\bibitem{bp-s-10}
R.~Baeza-Yates and P.~V. Poblete.
\newblock {Searching}.
\newblock {\em Algorithms and Theory of Computation Handbook}, pp.~2-1{--}2-16.
  Chapman {\&} Hall/CRC, 2010.

\bibitem{bt-cafm-06}
A.~Ben-Aroya and S.~Toledo.
\newblock {Competitive analysis of flash-memory algorithms}.
\newblock {\em Algorithms {--} ESA 2006: 14th Annual European Symposium,
  Zurich, Switzerland, September 11-13, 2006, Proceedings}, pp.~100{--}111.
  Springer, Lecture Notes in Computer Science 4168, 2006,
  \href{http://dx.doi.org/10.1007/11841036\_12}%
{doi:10.1007/11841036\_12}.

\bibitem{bcmv-ifm-03}
R.~Bez, E.~Camerlenghi, A.~Modelli, and A.~Visconti.
\newblock {Introduction to flash memory}.
\newblock {\em Proc. IEEE} 91(4):489{--}502, 2003,
  \href{http://dx.doi.org/10.1109/JPROC.2003.811702}%
{doi:10.1109/JPROC.2003.811702}.

\bibitem{chang2007efficient}
L.-P. Chang.
\newblock {On efficient wear leveling for large-scale flash-memory storage
  systems}.
\newblock {\em Proc. ACM Symp. on Applied Computing}, pp.~1126{--}1130, 2007,
  \href{http://dx.doi.org/10.1145/1244002.1244248}%
{doi:10.1145/1244002.1244248}.

\bibitem{chang2007endurance}
Y.-H. Chang, J.-W. Hsieh, and T.-W. Kuo.
\newblock {Endurance enhancement of flash-memory storage, systems: an efficient
  static wear leveling design}.
\newblock {\em Proc. 44th ACM/IEEE Design Automation Conf. (DAC '07)},
  pp.~212{--}217, 2007, \href{http://dx.doi.org/10.1145/1278480.1278533}%
{doi:10.1145/1278480.1278533}.

\bibitem{chk-ifw-10}
Y.-H. Chang, J.-W. Hsieh, and T.-W. Kuo.
\newblock {Improving flash wear-leveling by proactively moving static data}.
\newblock {\em IEEE Trans. Comput.} 59(1):53{--}65, 2010,
  \href{http://dx.doi.org/10.1109/TC.2009.134}%
{doi:10.1109/TC.2009.134}.

\bibitem{ChenGN11}
S.~Chen, P.~B. Gibbons, and S.~Nath.
\newblock {Rethinking database algorithms for phase change memory}.
\newblock {\em Proc. 5th Conf. on Innovative Data Systems Research (CIDR)},
  pp.~21{--}31, 2011,
  \url{http://www.cidrdb.org/cidr2011/Papers/CIDR11_Paper3.pdf}.

\bibitem{devroye2003cuckoo}
L.~Devroye and P.~Morin.
\newblock {Cuckoo hashing: further analysis}.
\newblock {\em Inform. Process. Lett.} 86(4):215{--}219, 2003,
  \href{http://dx.doi.org/10.1016/S0020-0190(02)00500-8}%
{doi:10.1016/S0020-0190(02)00500-8}.

\bibitem{DBLP:conf/icalp/DietzfelbingerGMMPR10}
M.~Dietzfelbinger, A.~Goerdt, M.~Mitzenmacher, A.~Montanari, R.~Pagh, and
  M.~Rink.
\newblock {Tight thresholds for cuckoo hashing via XORSAT}.
\newblock {\em Automata, Languages and Programming: 37th International
  Colloquium, ICALP 2010, Bordeaux, France, July 6-10, 2010, Proceedings, Part
  I}, pp.~213{--}225. Springer, Lecture Notes in Computer Science 6198, 2010,
  \href{http://dx.doi.org/10.1007/978-3-642-14165-2\_19}%
{doi:10.1007/978-3-642-14165-2\_19}.

\bibitem{Fotakis}
D.~Fotakis, R.~Pagh, P.~Sanders, and P.~G. Spirakis.
\newblock {Space efficient hash tables with worst case constant access time}.
\newblock {\em Theory of Computing Systems} 38(2):229{--}248, 2005,
  \href{http://dx.doi.org/10.1007/s00224-004-1195-x}%
{doi:10.1007/s00224-004-1195-x}.

\bibitem{DBLP:conf/approx/FriezeMM09}
A.~M. Frieze, P.~Melsted, and M.~Mitzenmacher.
\newblock {An analysis of random-walk cuckoo hashing}.
\newblock {\em Approximation, Randomization, and Combinatorial Optimization.
  Algorithms and Techniques: 12th International Workshop, APPROX 2009, and 13th
  International Workshop, RANDOM 2009, Berkeley, CA, USA, August 21-23, 2009,
  Proceedings}, pp.~490{--}503. Springer, Lecture Notes in Computer Science
  5687, 2009, \href{http://dx.doi.org/10.1007/978-3-642-03685-9\_37}%
{doi:10.1007/978-3-642-03685-9\_37}.

\bibitem{Frieze:2011:ARC:2078866.2078869}
A.~M. Frieze, P.~Melsted, and M.~Mitzenmacher.
\newblock {An analysis of random-walk cuckoo hashing}.
\newblock {\em SIAM J. Comput.} 40(2):291{--}308, 2011,
  \href{http://dx.doi.org/10.1137/090770928}%
{doi:10.1137/090770928}.

\bibitem{Grupp:2009:CFM:1669112.1669118}
L.~M. Grupp, A.~M. Caulfield, J.~Coburn, S.~Swanson, E.~Yaakobi, P.~H. Siegel,
  and J.~K. Wolf.
\newblock {Characterizing flash memory: anomalies, observations, and
  applications}.
\newblock {\em Proc. 42nd IEEE/ACM Int. Symp. on Microarchitecture (MICRO 42)},
  pp.~24{--}33, 2009, \href{http://dx.doi.org/10.1145/1669112.1669118}%
{doi:10.1145/1669112.1669118}.

\bibitem{UBIFS}
A.~Hunter.
\newblock {A brief introduction to the design of UBIFS}.
\newblock White paper, 2008,
  \url{http://www.linux-mtd.infradead.org/doc/ubifs_whitepaper.pdf}.

\bibitem{inr-otscc-92}
S.~Irani, M.~Naor, and R.~Rubinfeld.
\newblock {On the time and space complexity of computation using write-once
  memory or is pen really much worse than pencil?}
\newblock {\em Mathematical Systems Theory} 25(2):141{--}159, 1992,
  \href{http://dx.doi.org/10.1007/BF02835833}%
{doi:10.1007/BF02835833}.

\bibitem{KarLuc-JCAM-02}
M.~Karo{\'n}ski and T.~{\L}uczak.
\newblock {The phase transition in a random hypergraph}.
\newblock {\em J. Comput. Appl. Math.} 142(1):125{--}135, 2002,
  \href{http://dx.doi.org/10.1016/S0377-0427(01)00464-2}%
{doi:10.1016/S0377-0427(01)00464-2}.

\bibitem{DBLP:journals/siamcomp/KirschMW09}
A.~Kirsch, M.~Mitzenmacher, and U.~Wieder.
\newblock {More robust hashing: cuckoo Hashing with a Stash}.
\newblock {\em SIAM J. Comput.} 39(4):1543{--}1561, 2009,
  \href{http://dx.doi.org/10.1137/080728743}%
{doi:10.1137/080728743}.

\bibitem{Knuth:1998}
D.~E. Knuth.
\newblock {\em {The Art of Computer Programming, Volume 3: Sorting and
  Searching}}.
\newblock Addison Wesley, 2nd edition, 1998.

\bibitem{DBLP:conf/esa/LehmanP09}
E.~Lehman and R.~Panigrahy.
\newblock {3.5-way cuckoo hashing for the price of 2-and-a-bit}.
\newblock {\em Algorithms {--} ESA 2009: 17th Annual European Symposium,
  Copenhagen, Denmark, September 7-9, 2009, Proceedings}, pp.~671{--}681.
  Springer, Lecture Notes in Computer Science 5757, 2009,
  \href{http://dx.doi.org/10.1007/978-3-642-04128-0\_60}%
{doi:10.1007/978-3-642-04128-0\_60}.

\bibitem{mhash}
N.~Mavroyanopoulos and S.~Schumann.
\newblock {Mhash}.
\newblock Open-source software library, \url{http://mhash.sourceforge.net/}.

\bibitem{mrs-p2-01}
M.~Mitzenmacher, A.~W. Richa, and R.~Sitaraman.
\newblock {The power of two random choices: A survey of techniques and
  results}.
\newblock {\em Handbook of Randomized Computing}, vol.~1, pp.~255{--}312.
  Kluwer Academic Publishers, 2000.

\bibitem{mitzenmacher2005probability}
M.~Mitzenmacher and E.~Upfal.
\newblock {\em {Probability and Computing: Randomized Algorithms and
  Probabilistic Analysis}}.
\newblock Cambridge University Press, 2005.

\bibitem{pagh2004cuckoo}
R.~Pagh and F.~F. Rodler.
\newblock {Cuckoo hashing}.
\newblock {\em J. Algorithms} 51(2):122{--}144, 2004,
  \href{http://dx.doi.org/10.1016/j.jalgor.2003.12.002}%
{doi:10.1016/j.jalgor.2003.12.002}.

\bibitem{pboz-fmc-97}
P.~Pavan, R.~Bez, P.~Olivo, and E.~Zanoni.
\newblock {Flash memory cells{---}an overview}.
\newblock {\em Proc. IEEE} 85(8):1248{--}1271, 1997,
  \href{http://dx.doi.org/10.1109/5.622505}%
{doi:10.1109/5.622505}.

\bibitem{Qureshi:2009}
M.~K. Qureshi, J.~Karidis, M.~Franceschini, V.~Srinivasan, L.~Lastras, and
  B.~Abali.
\newblock {Enhancing lifetime and security of PCM-based main memory with
  start-gap wear leveling}.
\newblock {\em Proc. 42nd IEEE/ACM Int. Symp. on Microarchitecture (MICRO)},
  pp.~14{--}23, 2009, \href{http://dx.doi.org/10.1145/1669112.1669117}%
{doi:10.1145/1669112.1669117}.

\bibitem{SchSha-Comb-85}
J.~Schmidt-Pruzan and E.~Shamir.
\newblock {Component structure in the evolution of random hypergraphs}.
\newblock {\em Combinatorica} 5(1):81{--}94, 1985,
  \href{http://dx.doi.org/10.1007/BF02579445}%
{doi:10.1007/BF02579445}.

\bibitem{Sed-03-linear-probing}
R.~Sedgewick.
\newblock {\em {Algorithms in Java, Parts 1{--}4: Fundamentals, Data
  Structures, Sorting, and Searching}}.
\newblock Addison Wesley, 3rd edition, 2003, pp.~615{--}619.

\bibitem{moran}
L.~Turner.
\newblock {Dylan Moran Interview: On Music Loved {\&} Loathed}.
\newblock {\em The Quietus}, November 24 2009,
  \url{http://thequietus.com/articles/03283-dylan-moran-interview-on-music-lov%
ed-loathed}.

\bibitem{wet-pcm-10}
H.-S.~P. Wong, S.~Raoux, S.~Kim, J.~Liang, J.~P. Reifenberg, B.~Rajendran,
  M.~Asheghi, and K.~E. Goodson.
\newblock {Phase change memory}.
\newblock {\em Proc. IEEE} 98(12):2201{--}2227, 2010,
  \href{http://dx.doi.org/10.1109/JPROC.2010.2070050}%
{doi:10.1109/JPROC.2010.2070050}.

\bibitem{woodhouse2001jffs}
D.~Woodhouse.
\newblock {JFFS: the journalling flash file system}.
\newblock {\em Ottawa Linux Symposium}, 2001,
  \url{https://sourceware.org/jffs2/jffs2.pdf}.

\bibitem{Wu:1994}
M.~Wu and W.~Zwaenepoel.
\newblock {eNVy: a non-volatile, main memory storage system}.
\newblock {\em Proc. 6th Int. Conf. on Architectural Support for Programming
  Languages and Operating Systems (ASPLOS)}, pp.~86{--}97, 1994,
  \href{http://dx.doi.org/10.1145/195473.195506}%
{doi:10.1145/195473.195506}.

\bibitem{Zukowski:2006:AH:1140402.1140410}
M.~Zukowski, S.~H{\'e}man, and P.~Boncz.
\newblock {Architecture-conscious hashing}.
\newblock {\em Proc. 2nd Int. Worksh. on Data Management on New Hardware (DaMoN
  '06)}, 2006, \href{http://dx.doi.org/10.1145/1140402.1140410}%
{doi:10.1145/1140402.1140410}.

\end{thebibliography}
